\documentclass[runningheads]{llncs}
\usepackage{graphicx}
\usepackage{amsmath,amssymb,amsfonts}
\usepackage[linesnumbered,algoruled,boxed,ruled,vlined]{algorithm2e}
\usepackage{cite}
\usepackage{algorithmic}
\usepackage{textcomp}
\usepackage{xcolor}
\usepackage{subfig}
\usepackage{tikz,xcolor}
\usepackage{tkz-graph} 
\usetikzlibrary{shapes,arrows, automata}
\usetikzlibrary{positioning}
\usepackage{xargs}
\usepackage[colorinlistoftodos,prependcaption,textsize=tiny]{todonotes}

\newcommand{\y}{{\overline{y}}}
\newcommand{\x}{{\overline{x}}}
\newcommand{\Y}{{{Y}}}
\newcommand{\X}{{{X}}}

\newcommand{\trname}[1]{$^{[\boldsymbol{#1}]}$}

\newcommand{\calS}{{\cal S}}
\newcommand{\calM}{{\cal M}}

\begin{document}
\title{Mining Precise Test Oracle Modelled by FSM}
\titlerunning{Mining Precise Test Oracle Modelled by FSM}
\author{Omer Nguena Timo }
\authorrunning{O. Nguena Timo}
\institute{Universit\'e du Qu\'ebec en Outaouais\\
Campus de Saint-J\'er\^ome, QC, Canada \\
\email{omer.nguena-timo@uqo.ca}}
\maketitle              

\begin{abstract}
Precise test oracles for  reactive systems such as critical control systems and communication protocols can be modelled with deterministic finite state machines (FSMs). Among other roles, they serve in evaluating the correctness of systems under test. 
A great number of candidate precise oracles (shortly, candidates) can be produced at the system design phase due to  uncertainties, e.g., when  interpreting their requirements expressed in ambiguous natural languages. Selecting the proper candidate becomes challenging for an expert.   We propose a test-driven approach to assist experts in this selection task.  The approach uses a non deterministic FSM to represent the candidates, includes  the partitioning of the candidates into subsets of candidates via  Boolean encodings and requires the intervention of experts to select subsets.  We  perform an empirical evaluation  of the applicability of the proposed  approach.

\end{abstract}

\section{Introduction}\label{sec:intro}
Test oracles (simply called oracles) are usually used to evaluate the correctness of systems' responses to test data. In black-box testing approaches,  test data are usually generated from
machine-readable specifications which can also be used  in
automating the evaluation of responses and 
the production of verdicts on the presence of faults.  In  white-box testing approaches~\cite{Fraser:2013}, test data  serve to cover some artifacts during executions of a system and an expert which plays the role of the  oracle  evaluates the responses. Devising automated proper oracles is needed; however it is a tedious task which almost always requires the human expertise. Efforts are needed to facilitate this task~\cite{Barr2015,weyuker82} and to alleviate the intervention of experts in  recurrent test activities.  

\par Our work consider a typical conformance testing scenario~\cite{Lee96}, where an oracle is a deterministic finite state machine (DFSM).
However, uncertainty can occur in  devising oracles. E.g., it can 
be a consequence of misunderstanding or misinterpretation of 
requirements of systems often described with natural languages~\cite{fantechi1994assisting,brunello2019synthesis,fantechi2003applications}. As a result 
of the uncertainty, a set of candidate  oracles can be proposed. For example, machine learning-based  translation approaches~\cite{fantechi1994assisting,stahlberg2020neural} for reactive systems return the most likely DFSM, but the latter may be undesired due to decisions made by automated translation procedures. Instead, they could automatically return a set of candidate oracles of which the likelihood is above a certain threshold. On the other hand when a  candidate oracle  is available (e.g., it can be in the form of a Program under test), a set of its versions can be produced mutating it with operations mimicking the introduction or the correction of faults. Such a set can compactly  be represented by a non deterministic finite state machine (NFSM) thus modelling an imprecise oracle. The candidate oracles are called precise in the opposite of the imprecise oracle defining them. Devising an  oracle then consists in mining the proper candidate  from the imprecise oracle. 

\par In this paper we propose an approach to mining the proper oracle from an imprecise oracle represented with a NFSM.  An expert can answer queries related to the correctness of NFSM's responses. An answer can be either yes or no. Based on the answers,  the proper DFSM is automatically mined.
We assume that the proper oracle is not available to the expert and the expert might have limited  time resources for answering the queries. In this context, the expert cannot check the equivalence between a candidate  oracle  and the unavailable proper oracle; so, polynomial time active learning approaches inspired by $L^*$~\cite{DBLP:journals/iandc/Angluin87} are less adequate for devising the proper DFSM. In our approach, distinct responses to the same test data permit to distinguish between candidate oracles. Responses, as well as the corresponding test data,   are  automatically computed.   Our approach is iterative and applies the "divide and conquer" principle over a current  set of "good"  candidates. At each iteration step, the current candidate set is divided into a subset of "good"  candidates exhibiting "expected" responses to  test data and the complementary subset of "bad" ones.   The approach uses a Boolean encoding of the imprecise oracle; it takes advantage of the efficiency of constraint solvers to facilitate the search of good candidates.

\par The paper is organized as follows. The next section provides preliminary definitions. In Section~\ref{sec:problem}, we describe the oracle mining problem and introduce the steps of our solution to it. In Section~\ref{sec:encoding} we propose a Boolean encoding for an imprecise oracle and test-equivalent candidates; then we present the reduction of an imprecise oracle based on the selection of expected responses by experts. In section~\ref{sec:mining}, we propose a procedure for verifying the adequacy of a test data set for mining an oracle and a  mining procedure based on automatic generation of test data.  Experiments for promoting the  applicability of the approach are presented in Section~\ref{sec:experiments}. In section~\ref{sec:related-work}, we present the related work. We conclude our work in Section~\ref{sec:conclusion}. 

\section{Preliminaries}\label{sec:preliminaries}

A {\em Finite State Machine} (FSM) is a 5-tuple $\calS = (S, s^0, X, Y, T)$, where $S$ is a finite set of states with initial state $s^0$; $X$ and $Y$ are finite non-empty disjoint sets of inputs and outputs, respectively; $T \subseteq S \times \X \times \Y \times S$ is a transition relation and a tuple $(s, x, y, s') \in T$ is called a transition from $s$ to $s'$ with input $x$ and output $y$. The set of transitions from state $s$ is denoted by $T(s)$. $T(s,x)$ denotes the set of transitions in $T(s)$ with input $x$.   For a transition $t = (s, x, y, s')$, we define $src(t) = s$, $inp(t)=x$, $out(t)= y$ and $tgt(t) =s'$.
The set of uncertain transitions in an object $A$ is denoted by $Unctn(A)$.
Transition $t$ is {\em uncertain} if $|T(src(t), inp(t))| > 1$, i.e., several transitions  from the $src(t)$ have the same input as $t$; otherwise $t$ is {\em certain}. The number $U_{s,x} = |T(s, x)|$ is called the \textit{uncertainty degree} of state $s$ on input $x$. $U_\calS = max_{s\in S, x\in X} U_{s,x}$ defines the uncertainty degree of  $\calS$. We say that $\calS$ is \textit{deterministic} (DFSM) if it has no uncertain transition, otherwise it is non-deterministic (NFSM). In other words $U_\calS \leq 1$ if $\calS$ is deterministic. 
$\calS$ is {\em completely specified} (complete FSM) if for each tuple $(s, x) \in  S \times \X$ there exists transition $(s, x, y, s') \in T$.

\par An {\em execution of $\calS$ in $s$}, $e= t_1t_2\ldots t_n$ is a finite sequence  of transitions forming a path from $s$ in the state transition diagram of $\calS$, i.e., $src(t_1)=s$, $src(t_{i+1}) = tgt(t_i)$ for every $i=1...n-1$.  Execution $e$ is \textit{deterministic} if every $t_i$ is the only transition in $e$ that belongs to $T(src(t_i),inp(t_i))$, i.e., $e$  does not include several uncertain transitions from the same state with the same input. $e$ is simply called an execution of $\calS$ if $s=s^0$. $\calS$ is {\em initially connected}, if for any state $s' \in S$ there exists an execution of $\calS$ to $s'$. A DFSM has only deterministic executions, while an NFSM can have both.  A trace $\x/\y$ is a pair of an input sequence $\x$ and an output sequence $\y$, both of the same length. The trace of $e$ is $inp(t_1)inp(t_2)\ldots inp(t_n)/ out(t_1)out(t_2)\ldots out(t_n)$.
A trace of $\calS$  in $s$ is a trace of an execution of $\calS$ in $s$. Let $Tr_\calS(s)$ denote the set of all traces of $\calS$ in  $s$ and $Tr_\calS$ denote the
set of traces of $\calS$ in the initial state $s^0$. Given a sequence $\beta \in (\X\Y)^*$, the input (resp. output)
projection of $\beta$, denoted $\beta_{\downarrow X}$ (resp. $\beta_{\downarrow \Y}$), is a sequence obtained
from $\beta$ by erasing symbols in $Y$  (resp. $\X$); if $\beta$ is the trace of execution $e$, then $\beta_{\downarrow X} = inp(e)$ (resp. $\beta_{\downarrow \Y} = out(e)$) is called the input (resp. output) sequence of $e$ and we say that $out(e)$ is the \textit{response} of $\calS$ in $s$  to (the application of) input sequence $inp(e)$. $|X|$ denotes the size of set $X$.

\par Two complete FSMs are distinguished with an input sequence for which they produce different responses.  Given input sequence $\x \in \X^*$, let $out_\calS(s, \x)$ denote the set of responses which can be produced by $\calS$  when $\x$ is applied at state $s$, that is  $out_\calS(s, \x) = \{\beta_{\downarrow Y} \mid \beta\in Tr_\calS(s) \text{ and }\beta_{\downarrow X}= \x\}$.
Given state $s_1$ and $s_2$ of an FSM $\mathcal{S}$ and an input sequence $\x\in X^*$, $s_1$ and $s_2$ are $\x$-distinguishable, denoted by $s_1\not \simeq_\x s_2$ if  $out_\mathcal{S}(s_1,\x) \neq out_\mathcal{S}(s_2,\x)$; then $\x$ is called a distinguishing input sequence for $s_1$ and $s_2$. $s_1$ and $s_2$ are $\x$-equivalent, denoted by $s_1 \simeq_\x s_2$ if  $out_\mathcal{S}(s_1,\x) = out_\mathcal{S}(s_2,\x)$. $s_1$ and $s_2$ are distinguishable, denoted by $s_1\not \simeq s_2$,  if they are $\x$-distinguishable  for some input sequence $\x\in X^*$; otherwise they are equivalent. Let $a\in X$. A distinguishing input sequence $\x a \in X^+$ for $s_1$ and $s_2$ is {\em minimal} if $\x$ is not distinguishing for $s_1$ and $s_2$. 
Two complete DFSMs $\calS_1 = (S_1, s_1^0, X, Y, T_1)$ and $\calS_2 = (S_2, s_2^0, X, Y, T_2)$ over the same input and output alphabets are distinguished with input sequence $\x$ if $s_1^0\not \simeq_\x s_2^0$.

\par Henceforth, FSMs and DFSMs are complete and initially connected. 

\par Given a  NFSM $\calM  = (M, m^0, X, Y, N)$, a  FSM $\calS = (S, s^0, X, Y, T)$ is a {\em submachine} of $\calM$, denoted by $\calS \in \calM$ if $S\subseteq M$, $m^0 =s^0$
and $T\subseteq N$. 

\par We will use a  NFSM to represent a set of   \textit{candidate} DFSMs. We let $Dom(\mathcal{M})$ denote the set of candidate DFSMs included in NFSM $\mathcal{M}$. Later,  we will be interested in executions of $\mathcal{M}$ that are executions of a DFSM in $Dom(\mathcal{M})$.  Let $e$ be an execution of a NFSM $\calM$ in $m^0$. We say that  $e$ \textit{involves} a submachine $\calS = (S, s_0, X, Y, T)$ of $\calM$  if $Unctn(e) \subseteq T$, i.e., all the uncertain transitions in $e$ are defined in $\calS$. The certain transitions are defined in each DFSM in $Dom(\mathcal{M})$, but distinct DFSMs in $Dom(\mathcal{M})$ define distinct sets of  uncertain transitions.

\section{The Oracle Mining Problem and Overview of the Proposed Solution}\label{sec:problem}

Oracles play an important role in  testing and verification activities, especially
they define and evaluate the responses of implementations to given tests. The evaluation serves  to provide verdicts on the presence of faults in the implementations.
Letting experts play the role of an oracle is expensive. The experts will intervene in recurrent test campaigns  for judging an important number of responses. For these reasons, automated test oracles are preferred.

\par Devising  precise oracles (shortly oracles) is a challenging task that might require uncertainty resolution, as discussed in Section~\ref{sec:intro}. Full automation of this task might result in undesired oracles. Inspired by previous work~\cite{Chow78,MavridouL18}, we represent  oracles  with DFSMs  and a \textit{test} with an input sequence.

\par We propose a semi-automated mining approach  for devising  oracles. First we suggest modelling uncertainties with non deterministic transitions in a NFSM. This latter NFSM represents an \textit{imprecise oracle} and it defines  conflicting outputs for the same input applied in the same state.  It also defines a possibly big number of candidate  oracles (shortly \textit{candidates}) which are the DFSM included in it. Secondly, experts can take useful decisions  for the resolution of uncertainties and the automatic extraction of the proper candidate. The decisions concern the evaluation and the selection of conflicting responses. The fewer are the decisions, the less is  the intervention of experts in the mining process and the recurrent testing activities with the selected oracle.

\par Let a NFSM $\mathcal{M}  = (M, m^0, X, Y, N)$ represent an imprecise oracle. We say that $\mathcal{S} \in Dom(\mathcal{M})$ is the \textit{proper}  oracle w.r.t. experts if  $\mathcal{S}$ always produces the expected responses to every test, according to the point of view of experts; otherwise $\mathcal{S}$ is \textit{inappropriate}. Equivalent DFSMs represent an identical oracle. In practice the uncertainty degree of $\calM$ should be much smaller than its maximal value $|M||Y|$; we believe that it could be smaller than the maximum of $|M|$ and $|Y|$.
The \textit{oracle mining problem} is to select the proper  oracle in $\mathcal{M}$, with the help of an expert. We assume that $Dom(\mathcal{M})$ always contains the proper oracle.

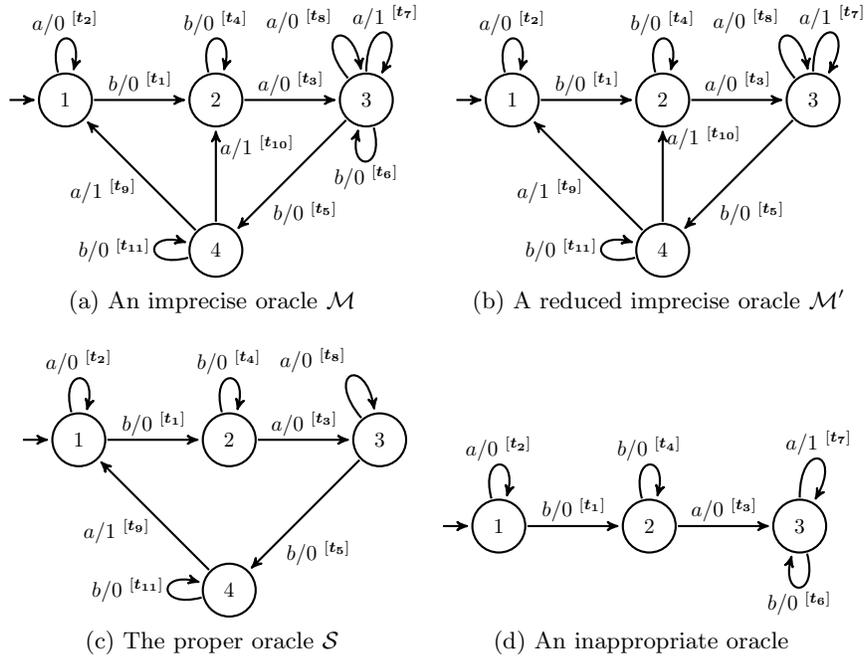
\begin{figure}[!ht]
\centering
\subfloat[An imprecise oracle $\mathcal{M}$\label{fig:imprecise}]{ 
\scalebox{0.80}{
\begin{tikzpicture}[->,>=stealth',shorten >=1pt,auto,node distance=2.5cm,initial text={},
every initial by arrow/.style={->}]
{\fontsize{10}{10}\selectfont
\tikzstyle{every node}=[scale=1,inner sep=1pt, outer sep=1pt,line width=1pt]
\tikzstyle{every edge}=[scale=1,line width=1pt, draw]
\tikzstyle{every state}=[fill=white,draw,text=black]
  \node[initial,state]  (S1)   at (0,5)                {$1$};
  \node[state]         (S2)  at ([xshift=2.5cm,yshift=0cm]S1) {$2$};
  \node[state]         (S3) at ([xshift=2.5cm,yshift=0cm]S2) {$3$};
  \node[state]         (S4) [below of=S2] {$4$};
  \path (S1) edge [loop above] node {$a/0$ \trname{t_2}} (S1)
             edge          node[above] {$b/0$ \trname{t_1}}(S2) 
        (S2) edge [loop above] node {$b/0$ \trname{t_4}} (S2)
             edge  node[above, pos=0.5] {$a/0$ \trname{t_3}} (S3)
        (S3)edge [loop, in=100, out=130,looseness=10] node {$a/0$ \trname{t_8} } (S3) 
             edge [loop, in=60, out=90,looseness=10]       node[xshift=-0.6cm] {$a/1$ \trname{t_7}} (S3)
             edge [loop below]       node {$b/0$ \trname{t_6}} (S3)
             edge node[pos=0.7] {$b/0$ \trname{t_5}} (S4)
        (S4) edge [loop left]  node {$b/0$ \trname{t_{11}}} (S4)          
             edge node {$a/1$ \trname{t_9}} (S1) 
             edge node[right, pos=0.8] {$a/1$ \trname{t_{10}}} (S2) ;
}
\end{tikzpicture}
}
} 
\subfloat[A reduced imprecise oracle $\mathcal{M}'$\label{fig:imprecise:red}]{ 
\scalebox{0.80}{
\begin{tikzpicture}[->,>=stealth',shorten >=1pt,auto,node distance=2.5cm,initial text={},
every initial by arrow/.style={->}]
{\fontsize{10}{10}\selectfont
\tikzstyle{every node}=[scale=1,inner sep=1pt, outer sep=1pt,line width=1pt]
\tikzstyle{every edge}=[scale=1,line width=1pt, draw]
\tikzstyle{every state}=[fill=white,draw,text=black]
  \node[initial,state]  (S1)   at (0,5)                {$1$};
  \node[state]         (S2)  at ([xshift=2.5cm,yshift=0cm]S1) {$2$};
  \node[state]         (S3) at ([xshift=2.5cm,yshift=0cm]S2) {$3$};
  \node[state]         (S4) [below of=S2] {$4$};
  \path (S1) edge [loop above] node {$a/0$ \trname{t_2}} (S1)
             edge          node[above] {$b/0$ \trname{t_1}}(S2) 
        (S2) edge [loop above] node {$b/0$ \trname{t_4}} (S2)
             edge  node[above, pos=0.5] {$a/0$ \trname{t_3}} (S3)
        (S3) edge [loop, in=100, out=130,looseness=10] node {$a/0$ \trname{t_8} } (S3) 
             edge [loop, in=60, out=90,looseness=10]       node[xshift=-0.6cm] {$a/1$ \trname{t_7}} (S3)
             edge node[pos=0.7] {$b/0$ \trname{t_5}} (S4)
        (S4) edge [loop left]  node {$b/0$ \trname{t_{11}}} (S4)          
             edge node {$a/1$ \trname{t_9}} (S1) 
             edge node[right, pos=0.9] {$a/1$ \trname{t_{10}}} (S2) ;
}
\end{tikzpicture}
}
}  

\subfloat[The proper  oracle $\mathcal{S}$\label{fig:precise1}]{
\scalebox{0.80}{
\begin{tikzpicture}[->,>=stealth',shorten >=1pt,auto,node distance=2.5cm,initial text={},
every initial by arrow/.style={->}]
{\fontsize{10}{10}\selectfont
\tikzstyle{every node}=[scale=1,inner sep=1pt, outer sep=1pt,line width=1pt]
\tikzstyle{every edge}=[scale=1,line width=1pt, draw]
\tikzstyle{every state}=[fill=white,draw,text=black]
  \node[initial,state]  (S1)   at (0,5)                {$1$};
  \node[state]         (S2)  at ([xshift=2.5cm,yshift=0cm]S1) {$2$};
  \node[state]         (S3) at ([xshift=2.5cm,yshift=0cm]S2) {$3$};
  \node[state]         (S4) [below of=S2] {$4$};
  \path (S1) edge [loop above] node {$a/0$ \trname{t_2}} (S1)
             edge          node[above] {$b/0$ \trname{t_1}}(S2) 
        (S2) edge [loop above] node {$b/0$ \trname{t_4}} (S2)
             edge  node[above, pos=0.5] {$a/0$ \trname{t_3}} (S3)
        (S3) edge [loop, in=100, out=130,looseness=10] node {$a/0$ \trname{t_8} } (S3)  
             edge node[pos=0.7] {$b/0$ \trname{t_5}} (S4)
        (S4) edge [loop left]  node {$b/0$ \trname{t_{11}}} (S4)          
             edge node {$a/1$ \trname{t_9}} (S1)  ;
}
\end{tikzpicture}
}
}
\subfloat[An inappropriate  oracle 
\label{fig:precise4}]{
\scalebox{0.80}{
\begin{tikzpicture}[->,>=stealth',shorten >=1pt,auto,node distance=2.5cm,initial text={},
every initial by arrow/.style={->}]
{\fontsize{10}{10}\selectfont
\tikzstyle{every node}=[scale=1,inner sep=1pt, outer sep=1pt,line width=1pt]
\tikzstyle{every edge}=[scale=1,line width=1pt, draw]
\tikzstyle{every state}=[fill=white,draw,text=black]
  \node[initial,state]  (S1)   at (0,5)                {$1$};
  \node[state]         (S2)  at ([xshift=2.5cm,yshift=0cm]S1) {$2$};
  \node[state]         (S3)  at ([xshift=2.5cm,yshift=0cm]S2) {$3$};
  \path (S1) edge [loop above] node {$a/0$ \trname{t_2}} (S1)
             edge          node[above] {$b/0$ \trname{t_1}}(S2) 
        (S2) edge [loop above] node {$b/0$ \trname{t_4}} (S2)
             edge  node[above, pos=0.5] {$a/0$ \trname{t_3}} (S3)
        (S3) edge [loop, in=60, out=90,looseness=10]       node[xshift=-0.6cm] {$a/1$ \trname{t_7}} (S3)
             edge [loop below]       node {$b/0$ \trname{t_6}} (S3);
}
\end{tikzpicture}
}
}

    \caption{An imprecise oracle and two plausible  oracles \label{fig:iut:mutmachine}}
\end{figure}

\par The NFSM in Figure~\ref{fig:imprecise} represents an imprecise oracle. It defines eight candidate oracles with six uncertain transitions, namely $t_5,t_6,t_7,t_8,t_9,t_{10}$. 
Figure~\ref{fig:precise1} and Figure~\ref{fig:precise4} present two candidates; one of them is proper.

\par Mining the proper  oracle is challenging even with the help of an expert, especially when the NFSM for an imprecise oracle defines an important number of candidates. The one-by-one enumeration of the candidates might not work because of the sheer number of candidates induced  by an imprecise oracle.  A naive approach could consist to deactivate in each state of the NFSM, the transitions producing  outputs evaluated as unexpected by the expert.
This naive approach does not work.  For example, the imprecise oracle in Figure~\ref{fig:iut:mutmachine} has four executions with input sequence $baba$, namely $t_1t_3t_5t_9$, $t_1t_3t_5t_{10}$, $t_1t_3t_6t_8$ and $t_1t_3t_6t_7$. The two plausible responses for these executions are $0000$ and $0001$. The latter is expected as it is produced  by the proper  oracle in Figure~\ref{fig:precise1}.  
\par All but one  executions produce the desired output $1$ in state 3 on the last input $a$. One could deactivate or remove the transition $t_8$ based on the fact that it produces the last undesired output in the unexpected response. In consequence the reduction of the imprecise oracle will result in an oracle not defining $t_8$. Any candidate not defining $t_8$ is not equivalent to the proper  oracle. This naive approach of selecting some transitions from  transition sequences fails in mining the proper  oracle. This is because entire sequences of transitions  used to reach  states (and so their input-output sequences) define the proper candidate.

\par Our oracle mining approach  relies on the evaluation by experts of responses (instead of isolated outputs) of the candidates to tests. The principle of the approach is iterative and quite simple. At each iteration step, first we use pair of candidates to generate tests. Next, we generate the plausible responses for generated tests. Then we let experts select expected responses. Eventually we remove from the candidate set, the ones producing unexpected responses;  this can be done by deactivating transitions in imprecise oracle and removing candidates from the set of solutions of the Boolean formulas. The iteration process continues if two remaining candidates are distinguishable.
A lot of memory can be needed to store each and every candidate, especially if a great number of them is available. To reduce the usage of the memory, we encode candidates with Boolean formulas and we use a solver to retrieve candidates from the Boolean encodings. The Boolean encoding is also useful for representing the  candidates already used to generate distinguishing tests.

\par In the next section we propose  Boolean encodings for the DFSMs including in a NFSM and the test-equivalent DFSMs. We also present how to deactivate/remove transitions in a NFSM for modelling  reduced candidate sets.

\section{Boolean Encodings}\label{sec:encoding}

Let $\mathcal{M} =(M,m_0, X, Y, T)$ be  an imprecise oracle. $Dom(\mathcal{M})$ represents a set of candidate  oracles, i.e., a set of DFSMs. We encode candidates with Boolean formulas over variables representing the transitions in $\mathcal{M}$. A solution of a formula \textit{determines} the transitions corresponding to the variables it assigns to "true". An FSM is \textit{determined} (encoded) by a formula if exactly all its transitions are determined by a solution of the formula.

\subsection{Candidates in an imprecise oracle}

Let  $\tau=\{t_1,t_2,\ldots, t_n\}$ be a set of variables, each variable corresponds to a transition in $T$. Let us define the Boolean expression $\xi_\tau$ as follows: 
$$
\xi_\tau =   \bigwedge_{k=1..n-1} (\neg t_k \vee \bigwedge_{j=k+1..n} \neg t_j)  \wedge \bigvee_{k=1..n} t_k
$$
It holds that every solution of  
$\xi_\tau$ determines exactly one variable in $\tau$. Indeed, $\xi_\tau$ assigns {\em True} if both $\bigwedge_{k=1..n-1} (\neg t_k \vee \bigwedge_{j=k+1..n} \neg t_j)$ and $\bigvee_{k=1..n} t_k$ are {\em True}.  
$\bigvee_{k=1..n} t_k$ is {\em True} whenever at least one $t_i$ is {\em True}. If  some $t_i$ is {\em True}, then every $t_j$, $i\neq j$ must be {\em False} in order for   $\bigwedge_{k=1..n-1} (\neg t_k \vee \bigwedge_{j=k+1..n} \neg t_j)$ to be {\em True}. 
So every solution of $\xi_\tau$ determines exactly one transition in $T$; this transition corresponds to the only variable in $\tau$ that the solution assigns to $True$.

\par We encode the candidates  in $Dom(\calM)$ with the formula $$\displaystyle \varphi_\mathcal{M} = \bigwedge_{(m,x)\in M\times X} \xi_{T(m,x)} 
$$
For every state $m\in M$ and every input $x\in X$, every solution of $\varphi_\mathcal{M}$ determines exactly one transition in $\mathcal{M}$, which entails that a solution of $\varphi_\mathcal{M}$ cannot determine two different transitions with the same input from the same state. So $\varphi_\mathcal{M}$ determines exactly the candidates in  $Dom(\varphi_\mathcal{M})$.

\par For the imprecise oracle  $\mathcal{M}$ in Figure~\ref{fig:imprecise}, $T(1,b) = \{t_1\}$,  $T(3, a) =\{t_7,t_8\}$, $\xi_{T(1, b)} = t_1$ and $\xi_{T(3, a)} = (\neg t_7 \vee \neg t_8) \wedge (t_7\vee t_8)$. Then, the formula $\varphi_\mathcal{M}:=$ 
$ t_1
\wedge  t_2
\wedge  t_3
\wedge  t_4
\wedge  t_{11}
\wedge  ((\neg t_7 \vee \neg t_8) \wedge (t_7\vee t_8))
\wedge  ((\neg t_5 \vee \neg t_6) \wedge (t_5\vee t_6))
\wedge  ((\neg t_9 \vee \neg t_{10}) \wedge (t_9\vee t_{10}))$ encodes all the DFSMs included in $\mathcal{M}$. In other words, $\varphi_\mathcal{M}$ determines all the candidates defined by $\mathcal{M}$. The DFSM in Figure~\ref{fig:precise1} is determined by $\varphi_{\mathcal{M}}$.

\subsection{Candidates involved in executions of an imprecise oracle}

\par An execution $e= t_1t_2\ldots t_n$ of $\mathcal{M}$ \textit{involves} a FSM $\calS \in  Dom(\mathcal{M})$ if every $t_i$ is defined in  $\calS$.  Recall that all the certain transitions are defined in every candidate. Let us define the formula $\varphi_e =  \bigwedge_{i=1..n, t_i \in Unctn(e)} t_i$. Clearly $\xi_e$ determines every uncertain transition in $e$, so it determines the deterministic and non deterministic FSMs involved in $e$. However we are interested in DFSMs in $Dom(\calM)$ only. Remark that if DFSM $\calS$ is involved in $e$, then $e$ is deterministic. Conversely, $e$ is deterministic if $Dom(\mathcal{M})$ includes  a DFSM involved in $e$. An execution of $\mathcal{M}$ must be deterministic for a DFSM to be involved in it. So $\varphi_e$ determines the DFSMs involved in $e$ if $e$ is deterministic. Let $E = \{e_1, e_2, \ldots, e_m\}$ be a set of deterministic executions of $\mathcal{M}$ and let us define the formula $\varphi_E = \bigvee_{i=1..n} \varphi_{e_i}$. The formula $\varphi_E \wedge \varphi_\mathcal{M}$ determines the DFSMs involved in an execution in $E$. 

\par Consider the NFSM in Figure~\ref{fig:imprecise} and a set $E =\{e_0= t_1t_3t_6t_8t_8t_6, e_1 = t_1t_3t_5t_9t_2, e_2 = t_1t_3t_5t_{10}t_3t_5, e_3=t_1t_3t_6t_7t_7t_6\}$ consisting of four executions $e_0,e_1,e_2$ and $e_3$.
Remark that the executions are deterministic and they have the same input sequence $babaab$ but distinct responses, namely $000000$ for $e_0$, $000100$ for $e_1$ and $e_2$ and $000110$ for $e_3$.  The formula $\varphi_E = ( t_1\wedge t_3 \wedge t_6 \wedge t_8) \vee (t_1 \wedge t_3 \wedge t_5 \wedge t_9 \wedge t_2) \vee (t_1 \wedge t_3 \wedge t_5 \wedge t_{10}) \vee (t_1 \wedge t_3 \wedge t_6 \wedge t_7) \wedge \varphi_\mathcal{M}$  encodes the DFSMs involved in the three executions. 

\subsection{Test-equivalent candidate}\label{sec:encoding:undistinguished}

\par Let $\x$ be a test. To determine the $\x$-equivalent DFSMs, we can partition $Dom(\mathcal{M})$ into subdomains. The DFSMs in each subdomain  produce the same response to test $\x$.  Our encoding of  each subdomain  with a Boolean formula works as follows. 

\par Let $\Y_{\mathcal{M},\x} =\{\y_1,\y_2,\ldots \y_n\}$ be the set of responses the DFSMs in $Dom(\mathcal{M})$  to test $\x$. Each response $\y_i$,  with $i=1...n$, corresponds a maximal set of deterministic executions of $\mathcal{M}$ with input sequence $\x$. We denote by $E_{\x/\y_i} =\{e_{i_1},e_{i_2},\ldots, e_{i_m}\}$ the set of deterministic executions producing $\y_i$ on input sequence $\x$. Clearly $E_{\x/\y_i}$ characterizes a subdomain of $\x$-equivalent DFSMs. The maximal size of $\Y_{\mathcal{M},\x}$ equals $|\x|^{|Y|}$ and it is reached when the imprecise oracle is the universe of all DFSMs, which is not the practical context of our work with imprecise oracles having reasonable uncertainty degrees.

\par Let $P_{x/\y_i}$ denote the set of DFSM in $\mathcal{M}$ involved in an execution in $E_{\x/\y_i}$. It holds that $P_{\x/\y_1}, P_{\x/\y_2}, \ldots P_{\x/\y_n}$ constitutes a partition of $Dom(\mathcal{M})$, i.e., every deterministic submachine of  $\mathcal{M}$ exactly belongs to one $P_{\x/\y_i}$, $i=1..n$ and every DFSM in  $P_{\x/\y_i}$ is a submachine of $\mathcal{M}$ for every $i=1..n$. 

\par For each $\y\in \Y_{\mathcal{M},\x}$, we define the formula $\varphi_{E_{\x/\y}}$. It holds that $\varphi_\calM \wedge \varphi_{E_{\x/\y}}$ encodes the maximal set of DFSMs indistinguishable by $\x$.  Indeed,  $\varphi_{E_{\x/\y}}$ determines exactly the $\x$-equivalent FSMs involved in deterministic executions in $E_{\x/\y}$ and $\varphi_\calM$ determines the DFSMs in $\calM$. We can show that  every DFSM included in $\mathcal{M}$ is determined by the formula $\varphi_\calM \wedge \varphi_{E_{\x/\y}}$ for exactly one $\y\in \Y_{\mathcal{M},\x}$.  Furthermore, if $\x$ is not distinguishing for the DFSMs in $Dom({\mathcal{M}})$, then $\varphi_\calM \wedge \varphi_{E_{\x/\y}}$ and $\varphi_\calM$ are equivalent, i.e., they determine  the DFSMs in  $Dom({\mathcal{M}})$.  

\par Considering our running example and the test $\x = babaab$, we have that $\Y_{\mathcal{M},babaab} = \{e_0= t_1t_3t_6t_8t_8t_6, e_1 = t_1t_3t_5t_9t_2, e_2 = t_1t_3t_5t_{10}t_3t_5, e_3=t_1t_3t_6t_7t_7t_6 \}$. Since the four executions have distinct responses (i.e., output sequences), we get $E_{babaab/000000} =\{e_0\}$, $E_{babaab/000100}=\{e_1, e_2\}$ and $E_{babaab/000110}=\{e_3\}$. %
Table~\ref{tab:subdomain:babaab} presents the corresponding subdomains and the number of  oracles in each subdomain. The two  oracles in the subdomain for response $000000$ are equivalent. The same for response $000110$. The subdomain for response $000100$ defines four $babaab$-equivalent  candidate oracles. Later, experts are invited to select the expected response that  will serve to reduce the imprecise oracle. 

\begin{table}[!t]
    \centering
    \caption{Partitioning of  $\mathcal{M}$ into Subdomains w.r.t input sequence $\x = babaab$}
    \label{tab:subdomain:babaab}
    \begin{tabular}{p{2cm}| p{4cm}|c|p{4cm}}
     \textbf{Response $\y$} & \textbf{Subdomain for $\varphi_{\cal M}$}   & \textbf{size} & \textbf{Precise oracles in the subdomain $P_{\x/\y_i}$}\\ \hline 
    $000100$  & $\varphi_{\x/000100} = ((t_5\wedge t_9) \vee (t_5\wedge t_{10}))$ & 4 & $\{t_1,t_2,t_3,t_4,t_5,t_7,t_{10},t_{11}\}$, $\{t_1,t_2,t_3,t_4,t_5,t_7,t_{9},t_{11}\}$, $\{t_1,t_2,t_3,t_4,t_5,t_8,t_{9},t_{11}\}$, $\{t_1,t_2,t_3,t_4,t_5,t_8,t_{10},t_{11}\}$ \\ \hline 
     $000110$ & $\varphi_{\x/000110} = t_6 \wedge t_7 $ & 2 & $\{t_1,t_2,t_3,t_4,t_6,t_7,t_{10},t_{11}\}$, $\{t_1,t_2,t_3,t_4,t_6,t_7,t_{9},t_{11}\}$\\ \hline 
     $000000$ & $\varphi_{\x/000000} = t_6 \wedge t_8$ & 2 & $\{t_1,t_2,t_3,t_4,t_6,t_8,t_{9},t_{11}\}$, $\{t_1,t_2,t_3,t_4,t_6,t_8,t_{10},t_{11}\}$ \\ \hline 
    \end{tabular}
    
    where, 
 \scriptsize{$\varphi_\mathcal{M}=$ 
$ t_1
\wedge  t_2
\wedge  t_3
\wedge  t_4
\wedge  t_{11}
\wedge  ((\neg t_7 \vee \neg t_8) \wedge (t_7\vee t_8))
\wedge  ((\neg t_5 \vee \neg t_6) \wedge (t_5\vee t_6))
\wedge  ((\neg t_9 \vee \neg t_{10}) \wedge (t_9\vee t_{10}))$}  
\end{table}

\subsection{Reducing an imprecise oracle}
\label{sec:encoding:reduced}

The selection of test-equivalent candidates  renders useless transitions of the imprecise oracle unused in the selected  candidates. These transitions can be deactivated for obtaining a reduced imprecise oracle.  

\par Let $\calM  = (M, m^0, X, Y, N)$ be an input complete NFSM and $\x/\y$ be a trace. $Dom(\mathcal{M})$ is partitioned into the set $Dom(\mathcal{M})_{\x/\y}$ of DFSMs producing $\y$ on $\x$ and the set of DFSMs not producing $\y$ on $\x$. We say that a transition $t\in N$ is \textit{eligible} for a candidate involved in $e$ if $e$ uses $t$  or $t'\not\in N(src(t),inp(t))$ for every $t'$ used in $e$.  

\begin{lemma}There is a submachine  $\mathcal{M}_{\x/\y}$ of $\mathcal{M}$ such that $Dom(\mathcal{M}_{\x/\y}) =  Dom(\mathcal{M})_{\x/\y}$. 
\end{lemma}
\begin{proof}
Let e be a deterministic execution $e$ in $E_{\x/\y}$.  Remark that all the transitions in $e$ are eligible for the candidates involved in $e$. Moreover $e$ is the only execution with input sequence $\x$ and response $\y$ in each of these candidates.  

\par We build $\mathcal{M}_{\x/\y}= (S, s^0, X, Y, T)$ with $T\subseteq N$ by deactivating  (deleting) non eligible transitions for candidates in $Dom(\calM_{\x/\y})$.  Formally $t\in N$ belongs to $T$ if it is eligible for a candidate involved in some deterministic execution $e \in E_{\x/\y}$. $m\in M$  belongs to $S$ if $m$ is used in a transition in $T$. 
Clearly, $\mathcal{M}_{\x/\y}$ is a complete and initially connected submachine of $\calM$; $\mathcal{M}_{\x/\y}$ is not necessarily deterministic because several executions in $E_{\x/\y}$ can use several uncertain transitions defined in the same state and with the same input; these transitions belong to $T$.  

\par First we show that $Dom(\mathcal{M}_{\x/\y}) \subseteq Dom(\mathcal{M})_{\x/\y}$ by contradiction. Assume that there is $\mathcal{P}$ in $ Dom(\mathcal{M}_{\x/\y})$ but not in $Dom(\mathcal{M})_{\x/\y}$. $\mathcal{P}$ is deterministic and by construction it defines all the transitions in a deterministic execution $e \in E_{\x/\y}$ of $\calM$. This implies the response of $\mathcal{P}$ on $\x$ is $\y$, which is a contradiction with  hypothesis  $\mathcal{P} \not \in Dom(\mathcal{M})_{\x/\y}$. 
Secondly, we show that $Dom(\mathcal{M})_{\x/\y} \subseteq Dom(\mathcal{M}_{\x/\y})$.
Let $\mathcal{P} \in Dom(\mathcal{M})_{\x/\y}$. $\mathcal{P}$ produces $\y$ on $\x$ with exactly one of its execution $e$. The transitions eligible for $\mathcal{P}$ are defined in $\mathcal{M}_{\x/\y}$. So $\mathcal{P} \in Dom(\mathcal{M}_{\x/\y})$. \qed 
\end{proof}
 
 \par Consider Table~\ref{tab:subdomain:babaab} and assume experts choose the expected response $000100$. The reduced imprecise oracle for $babaab/000100$, $\calM_{babaab/000100}$  is the imprecise oracle in Figure~\ref{fig:imprecise:red} which  was obtained  by removing transition $t_6$ from $\calM$ in  Figure~\ref{fig:imprecise}. This is because among the two transitions $t_5$ and $t_6$ from state $3$ with input $b$, the   executions in $E_{babaab/000100}$ only use  $t_5$.

\par Reducing an imprecise oracle permits to  speed up the computation of executions with given tests. Indeed, once it becomes clear that passing some transitions in the imprecise oracle leads to the production of undesired responses, one does not need to consider these transitions in determining new execution sets.

Let $\calS$ be a candidate in $Dom(\calM)$ and $\x/\y$ be a test-response pair. 
\begin{lemma}
$\mathcal{S} \in Dom(\calM_{\x/\y})$ if and only if  $\mathcal{S}$ is determined by $\varphi_\calM \wedge \varphi_{E_{\x/\y}}$. 
\end{lemma}

Remark that in some circumstances $\calM_{\x/\y}$ is the same as $\calM$. This happens when the union of eligible transitions over a set of executions equals the set of transitions of $\calM$. Such a case will be presented in Section~\ref{sec:generation}. Uncertain transitions in $\calM$ but not in $\calM_{\x/\y}$ are not determined by $\varphi_\calM \wedge \varphi_{E_{\x/\y}}$ because other uncertain transitions are determined by $\varphi_{E_{\x/\y}}$ and a solution of $\varphi_\calM$ cannot determine two uncertain transitions from the same state with the same input.

\section{Mining an  Oracle} \label{sec:mining}
To mine an oracle represented with a DFSM, we apply a test set $TS$ on an imprecise oracle $\mathcal{M}$. We say that $TS$  is \textit{adequate} for mining the proper oracle from  $\mathcal{M}$ if $TS$ is distinguishing for some $\mathcal{S} \in \mathcal{M}$ and every other candidate in $\mathcal{M}$ that is not equivalent to $\mathcal{S}$; moreover $\mathcal{S}$  is proper. Verifying the mining adequacy of $TS$ is the first step in mining the proper oracle. In case $TS$ is not adequate, new tests can be generated.

\subsection{Verifying adequacy of a test set for mining the proper oracle}\label{sec:verifying}
\begin{algorithm}[t]
   \SetKwInOut{KwInOut}{Input-Output}
    \SetKwProg{Fn}{Procedure}{:}{}
    \SetKwFunction{Verify}{\textit{verify\_test\_adequacy\_for\_mining}}
    \KwInOut{ $\mathcal{M}$ an imprecise oracle }
    \KwIn{ $\varphi_{\mathcal{M}}$ the boolean encoding of DFSM included in NFSM  $\mathcal{M}$ }
    \KwIn{a test set $\mathcal{TS}$}
    \KwIn{a DFSM  $\mathcal{S}$ emulating the expert for the response selection}
    \KwOut{$verdict$, is $true$ or $false$ on whether $\mathcal{TS}$ enables mining a DFSM.}
    \KwOut{$\varphi$ the Boolean encoding of DFSM consistent with expert knowledge}
    \KwOut{$\x_d$ a test that distinguish two DFSM}
        
    \Fn{\Verify{$\mathcal{M},\varphi_{\mathcal{M}}, \mathcal{TS}, \mathcal{S}$}}{
     
     Set $\varphi = \varphi_{\mathcal{M}}$
     
     Set $verdict = true$ if $\varphi$ does not select at least two non equivalent DFSMs; otherwise set $verdict =false$
     
     \While{$\mathcal{TS}\neq \emptyset$  and $verdict==false$}{
       Let $\x$ be a test in  $\mathcal{TS}$.
       
       Remove $\x$ from $\mathcal{TS}$.
       
       Determine $Y_{\mathcal{M},\x}$ the set of outputs of deterministic executions in $E_\x$ of $\mathcal{M}$ with input  $\x$ 
       
       Show $Y_{\mathcal{M},\x}$ to experts and let $\y \in Y_{\mathcal{M},\x}$ be the output such that $\y =out_{ \mathcal{S}}(s^0,\x)$, ($\to$ choice of the expected response by experts)

      Determine $E_{\x/\y}\subseteq E_\x$ the deterministic executions of  $\mathcal{M}$ which produce $\y$ on test $\x$ 
       
      Determine $\calM_{\x/\y}$

     Set $\varphi = \varphi \wedge \varphi_{E_{\x/\y}}  $ the Boolean encoding of DFSMs in $\mathcal{M}$ which produce $\y$ on test $\x$ 
      
     Set $\calM = \calM_{\x/\y}$

     \eIf{$\varphi$  encodes at  two non equivalent DFSMs}{
          Set $\x_d$ to a minimal distinguishing test for two non equivalent DFSMs
          }
          {
                     
               Set $verdict = true$
     }

     }
       
        \Return $(verdict, \mathcal{M}, \varphi, \x_d)$
    }
\caption{Verifying Test Adequacy For Mining an  Oracle.\label{algo:verifying}}
\end{algorithm}

\par Our method of verifying the adequacy of a test is iterative. At each iteration step, a test is randomly chosen and the corresponding plausible responses are computed with the imprecise oracle. Then experts select an expected response and send it to an automated procedure. The automated procedure reduces the imprecise oracle, i.e., deactivates some transitions from the imprecise oracle. The procedure stops when the responses for every test are examined or no imprecision remains. The procedure \textit{verify\_test\_adequacy\_for\_mining} scripted in Algorithm~\ref{algo:verifying} returns a verdict of the verification. 

\par Procedure \textit{verify\_test\_adequacy\_for\_mining}  takes as inputs an imprecise oracle represented by a NFSM, a test set and the expert knowledge about the expected outputs for the tests. We represent the expert knowledge with a DFSM.  It uses Boolean encoding presented in the previous section. The procedure ends the iteration if all the tests were visited or the Boolean encoding defines a single DFSM. If the Boolean encoding of the test-equivalent DFSMs  defines two non equivalent DFSMs then the tests do not enable mining an  oracle; otherwise one of the remaining equivalent  DFSMs is mined. The procedure also returns the Boolean encoding of the selected DFSMs for the tests, i.e, the DFSMs which produce the expected output on every test.

\par Consider the original imprecise oracle $\calM$ in Figure~\ref{fig:imprecise}. For verifying whether  the test $babaab$ is adequate for mining an  oracle, \textit{verify\_test\_adequacy\_for\_mining} determines the plausible responses (see  Table~\ref{tab:subdomain:babaab}) for the deterministic execution  
$\calM$ on $babaab$.  Assume that experts choose expected response $000100$. The procedure determines $E_{babaab/000100}$ as we discussed in Section~\ref{sec:encoding:undistinguished}; then it builds $\varphi_{babaab/000100}$ in Table~\ref{tab:subdomain:babaab} and the reduced imprecise oracle in Figure~\ref{fig:imprecise:red} as  discussed in Section~\ref{sec:encoding:reduced}.
The  formula $\varphi:= \varphi_{\calM} \wedge \varphi_{babaab/000100}$
determines four $babaab$-equivalent candidates presented in Table~\ref{tab:subdomain:babaab}.  Two of these candidates are distinguished with test $babaaa$, namely the  oracle in Figure~\ref{fig:precise1} and the one defining the transition set $\{t_1, t_2, t_3, t_4, t_5, t_7, t_{10}, t_{11}\}$. This latter  oracle provides response $000101$  whereas the former provides $000100$ for test $babaaa$.
In conclusion the procedure returns $verdict =false$ indicating that test $babaab$ is not adequate for mining the proper  oracle in Figure~\ref{fig:precise1}; it also returns the reduced imprecise oracle and the encoding with  $\varphi$ of $babaab$-equivalent candidates.

\subsection{Test generation in mining an  oracle}\label{sec:generation}

\begin{algorithm}[t]
 \SetKwProg{Fn}{Procedure}{:}{}
    \SetKwFunction{mining}{\textit{precise\_oracle\_mining}}
    \KwIn{$\varphi_{\mathcal{M}}$ the boolean encoding of DFSM included in a NFSM  $\mathcal{M}$}
    \KwIn{a test set $\mathcal{TS}$}
    \KwIn{a DFSM  $\mathcal{S}$ emulating the expert for the response selection}
    \KwOut{$\mathcal{TS}_m$ a test set that enables mining a DFSM.}
    \KwOut{$\mathcal{P}$ the proper oracle}    
    
    \Fn{\mining{$\mathcal{M}, \mathcal{TS}, \mathcal{S}$}}{
    
    Set $\varphi = \varphi_{\mathcal{M}}$
    
    Set $\mathcal{TS}_m = \mathcal{TS}$
    
    $(verdict, \mathcal{M}', \varphi', \x_d) = \textit{verify\_test\_adequacy\_for\_mining}(\mathcal{M},\varphi, \mathcal{TS}, \mathcal{S})$
    
    \While{$verdict==false$}{
    
     Set $\mathcal{TS}_m = \mathcal{TS}_m \cup \{\x_d\}$
     
    $\varphi = \varphi'$
    
    $\mathcal{M} = \mathcal{M}'$
    
    Set $\mathcal{TS} = \{\x_d\}$
    
   $(verdict, \mathcal{M}', \varphi', \x_d) = \textit{verify\_test\_adequacy\_for\_mining}(\mathcal{M},\varphi, \mathcal{TS}, \mathcal{S})$
    }
    
    Let $\mathcal{P}$ be the DFSM obtained from a solution of $\varphi'$
    
    \Return $(\mathcal{TS}_m, \mathcal{P})$
    
    }
   
\caption{Mining an Oracle by Test Generation.\label{algo:mining}}
\end{algorithm}

\par Procedure \textit{precise\_oracle\_mining} in Algorithm~\ref{algo:mining}  mines an  oracle from an imprecise one by generating tests.  The procedure makes a call to semi-automated procedure \textit{verify\_test\_adequacy\_for\_mining} in Algorithm~\ref{algo:verifying}. If given tests are not adequate for the mining task, procedure  \textit{verify\_test\_adequacy\_for\_mining} returns a Boolean encoding of a reduced set of test-equivalent candidates. Then, procedure \textit{precise\_oracle\_mining} generates a distinguishing test for two candidates in the reduced set.
Such a test can correspond to  a path to a sink state in the distinguishing product~\cite{PetrenkoNR16} of two  candidates. 
The test generation stops if the  generated test is adequate for mining the proper  oracle in the reduced set of candidates;  otherwise another test is generated. 
Procedure \textit{precise\_oracle\_mining} always terminates because at each iteration step, the set of  candidates is reduced after a call to procedure \textit{verify\_test\_adequacy\_for\_mining} and the number of DFSMs included in the original imprecise oracle is finite. On termination of \textit{verify\_test\_adequacy\_for\_mining}, the initial tests augmented with the generated ones constitute adequate tests for mining the proper oracle  determined by $\varphi'$.   

\par Considering the running example, the first call to  \textit{verify\_test\_adequacy\_for\_mining} in the execution of Procedure \textit{precise\_oracle\_mining} permits establishing that the test $babaab$ is not adequate for mining an  oracle. This was discussed at the end of the previous section where the test $\x_d = babaaa$ was generated as a distinguishing test for two candidates determined by  $\varphi':= \varphi_{\calM} \wedge \varphi_{babaab/000100}$ and included in the reduced imprecise oracle $\calM'$ in Figure~\ref{fig:imprecise:red}. In the first iteration step of the while loop, Procedure  \textit{precise\_oracle\_mining} makes a second call to \textit{verify\_test\_adequacy\_for\_mining} for checking whether the generated test $babaaa$ is adequate for mining an oracle from the new context $\calM = \calM'$ and $\varphi = \varphi'$. Here is what happens within this second call. The plausible responses for $babaaa$ belong to $Y_{\calM', babaaa} =\{000100,000101\}$; they are obtained with  deterministic executions of $\calM'$ in $E_{babaaa} =\{e_0= t_1t_3t_5t_9t_2t_2, e_1= t_1t_3t_5t_{10}t_3t_8, e_2= t_1t_3t_5t_{10}t_3t_7 \}$. Computing executions having input sequence $babaaa$ and the plausible responses is more efficient with $\calM'$ than with $\calM$; this is because $\calM'$  does not define $t_6$. Assume that $000100$ is the expected response for $babaaa$. Then $E_{babaaa/000100} =\{e_0= e_1=t_1t_3t_5t_9t_2t_2, e_2= t_1t_3t_5t_{10}t_3t_8\}$ and  $\varphi_{babaaa/000100} = t_9 \vee (t_{10} \wedge t_8)$.  Using $\calM'$ in Figure~\ref{fig:imprecise:red}, there are two candidates involved in $e_0$ and the eligible transitions for the two candidates  include all the transitions in $\calM'$ but $t_{10}$. Remark that uncertain transitions $t_8$, $t_7$ are eligible even if they are not used in $e_0$. There is one candidate involved in $e_1$ and the eligible uncertain transitions for this candidate are  $t_8$, $t_{10}$. So, the set of eligible transitions for the candidates involved in executions in $E_{babaaa/000100}$ are all the transitions in $\calM'$. In this particular case, $\calM'$ is not reduced with test-response pair $babaaa/000100$. However the $\{babaab,babaaa\}$-equivalent candidates  are encoded with $\varphi' \wedge \varphi_{E_{babaaa/000100}} = \varphi_\calM \wedge \varphi_{E_{babaab/000100}} \wedge\varphi_{E_{babaaa/000100}}$. This latter formula determines two candidates distinguishable with $babaaba$ in the reduced imprecise oracle obtained from $\calM'$ by deactivating transition $t_{10}$. Eventually  \textit{precise\_oracle\_mining} generates the test $baa$, terminates and returns adequate  test set $\{babaab,babaaa, babaaba, baa\}$ for mining the  oracle in Figure~\ref{fig:precise1}.

\section{Experimental Results }\label{sec:experiments}

We evaluate whether  the proposed approach is applicable for mining oracles from  imprecise oracles that define a big number of candidate  oracles and whether it requires a reasonable number of interventions of  experts. For that purpose we implemented a prototype tool, perform  multiple atomic experiments,  monitor  metrics and we compute some statistics.  
The prototype tool is implemented in Java; it uses Java libraries of the solver Z3 version 4.8.4 and the compilation tool ANTLR  version 4.7.2. The  computer has the following settings: WINDOWS 10, 16 Go (RAM), Intel(R) Core i7-3770 @ 3.4 GHz.  

\par An atomic experiment works as follows. We automatically generate a complete DFSM ${\cal S}$ for given numbers of  states, inputs and outputs denoted by $|M|, |X|$ and $|Y|$ respectively. ${\cal S}$ emulates the experts during the experiments. We  set  the uncertainty degree $U$. For a  value of $U$ we randomly add transitions to ${\cal S}$ for generating an imprecise oracle ${\cal M}$. Eventually, we extract a DFSM equivalent to ${\cal S}$ from ${\cal M}$ by making a call to our implementation of  procedure \textit{precise\_oracle\_mining} in Algorithm~\ref{algo:mining}. 

\par The metrics we monitor in each atomic experiments are:  $|Dom({\cal M})|$ the maximum number of candidate  oracles in $\mathcal{M}$; $|TS|_{min}$ and $|TS|_{max}$ the minimum and the maximum numbers of generated tests;  $L_{min}$ and $L_{max}$ the minimum and the maximum lengths of the generated tests; and $T_{min}$, $T_{max}$ and $T_{med}$ the minimal, maximal and median processing times (in milliseconds) for the mining procedure. We assumed that it takes almost zero millisecond for  emulated experts to select responses, which is insignificant in comparison to the processing time for the plausible responses and  solutions of Boolean formulas.
We performed 30 atomic experiments to obtain the data in each row of Table~\ref{tab:result:scalability0} and
Table~\ref{tab:result:scalability1}.

\par  In Table~\ref{tab:result:scalability0}, we consider imprecise oracles with $10$ states, $3$ inputs and $2$ outputs. We observe that the values of almost all the metrics augment when the uncertainty degree $U$ increases, especially   $T_{med}$. 
The generated imprecise oracles in Table ~\ref{tab:result:scalability1} have  $3$ inputs, $2$ outputs and uncertainty degree equals to $3$. We also observe that almost all the metrics increase when the number of states increases, especially   $T_{med}$. We notice that for $(|M|, |X|, |Y|, U) = (10,3,2,3)$, the gap between the values for $T_{med}$ in Table ~\ref{tab:result:scalability0} and Table ~\ref{tab:result:scalability1} is minor, which let us believe that $T_{med}$ is significant to evaluate the performance of our  approach.

\par Let us provide a practical perspective on the results in Table~\ref{tab:result:scalability0} and Table~\ref{tab:result:scalability1}. Clearly, experts would have took more time than its emulation with a DFSM  to select expected responses. Let us assume that it takes on average 1 minute to experts for selecting the expected response for a test. Under this assumption and considering the last row of Table ~\ref{tab:result:scalability0}, the extraction of an  oracle over the possible $2.21E23$ candidates could last $106$ minutes since the automated procedure only lasts for $18.26$ seconds. We advocate that if the extracted  oracle serve in testing a critical system, taking $106$ minutes to extract the proper  oracle is better than using an undesired  oracle. If the manual repair of the undesired oracle is not trivial, mutation operations (taking inspiration from ~\cite{Weimer:2009,LeGoues:2012}) can apply to it for generating an imprecise oracle and mining a proper oracle.  

\begin{table}[!t]
 \centering
\caption{$(|M|,|X|,|Y|) = (10,3,2)$} \label{tab:result:scalability0}
\begin{tabular}{|c|c|c|c|c|c|c|c|c|}
$U$ & $|Dom(\mathcal{M})|$ & $|TS|_{min}$ & $|TS|_{max}$ & $L_{min}$ & $L_{max}$ & $T_{min} (ms)$ & $T_{max}(ms)$ & $T_{med}(ms)$ \\\hline \hline
2 &	1.07E9 & 21	& 32 &	5 &	8 &	871 &	1619 &	1106.0 \\
3 &	2.06E14	&33	& 55	& 5 &	8 &	2128	& 115867	& 2865.0 \\
4 &	1.15E18	&	40	& 78	& 5 &	7 &	3313 &	8626 &	4417.0 \\
5 &	9.31E20 &	55 &	100	& 5	& 7 &	6334 &	35190 &	9618.0 \\
6 &  2.21E23 &	64	& 106 &	5 &	7 &	9903 &	105994 &	18263.0 \\
\end{tabular}
\end{table}

 \begin{table}[!t]
 \centering
\caption{$(|X|,|Y|, U) = (3,2,3)$} \label{tab:result:scalability1}
\begin{tabular}{|c|c|c|c|c|c|c|c|c|}
$|M|$ & $|Dom(\mathcal{M})|$ & $|TS|_{min}$ & $|TS|_{max}$ & $L_{min}$ & $L_{max}$ & $T_{min}(ms)$ & $T_{max}(ms)$ & $T_{med}(ms)$ \\\hline \hline
7	&	1.05E10	&	22 &	43 &	4 &	7&	1008&	2457&	1220.0 \\
8 & 2.82E11	&24	& 53 &	4 &	8 &	1136 &	3199 &	2071.0 \\
9 & 7.63E12	 & 30 &	55 &	5 &	7 &	1575 &	4767 &	2056.0 \\
10 & 2.06E14 &	33 &	53 &	5 &	7 &	1905 &	4237 &	2438.0 \\
11 & 5.56E15 &	37 &	66 &	5 &	7 &	2109 &	4567 &	3053.0 \\
12 & 1.50E17 &	41 &	71 &	5 &	8 &	2533 &	5588 &	5140.0 \\
13 & 4.053E18 &	43 &	79 &	5 &	8 &	2837 &	7680 &	6381.0\\
\end{tabular}
\end{table}

The proposed approach could also be lifted for the generation in a distributed way of  adequate test sets for mining each and every candidate. This can be done by partitioning the candidate set into subsets, one subset per plausible response.  The constraints for each subset  can be processed in parallel in other to generate new tests. The generated test sets  will be computed without any intervention of experts. After the test set generation and the iterative partitioning of candidate subsets, the experts could passively  select expected responses for the generated tests in a passive manner for  mining the proper  oracle.

\section{Related Work}\label{sec:related-work}
Metamorphic testing~\cite{chen2020metamorphic,Segura2016,SahaK20} applies in devising test oracle when it is difficult to compare an expected response of a system under test with an observed one.  It consists in mutating  original test input data  to build a test set that violates metamorphic relations. These relations can play the role of coarse specifications and can serve to derive test sets. Building the relations requires the expert knowledge and extra-skills. Our approach exonerates testers to building such relations. Candidate oracles  allow focusing on revealing  deviations in the  responses.
\par In~\cite{Weimer:2009,LeGoues:2012} a test-response set is used to repair a system when its formal specification is unavailable.  The approach consists in analyzing  mutated versions of an implementation (C program) until one is found that  retains required functionality and  avoids a defect located by the tests. Mutated versions are generated using genetic programming. In our work, the specification and the test-response pairs are unavailable. We  generate tests and we rely on experts and the imprecise oracle to obtain the expected responses and to extract  the oracle (specification).

In~\cite{Hieron2004}, a test set is generated to detect whether a DFSM implementation is a reduction (i.e., is trace included) of a NFSM specification playing the role of an oracle; if so the implementation conforms to the specification. This work presumes that any of the traces of the specification is expected. This differs from our settings where responses from non deterministic executions in the imprecise oracle NFSM cannot be produced by the proper candidate DFSM;  so any implementation exhibiting these responses must fail the tests.   

The work in~\cite{DBLP:journals/iandc/Angluin87} addresses the problem of learning a  DFSM by using output and equivalence queries to a teacher. The proposed polynomial time active learning algorithm often requires a certain number of queries so that it wont be effective for experts to play the role of the teacher. In practice, the teacher is a black-boxed implementation one wants to infer a DFSM model. In our work, we want to mine a DFSM from a given NFSM by using the expert knowledge.  Such a situation happens, e.g., when one needs to choose among multiple implementation models of  the same system. In our settings, there is no  equivalence query and expert responds few queries on the selection of expected responses. 

The work in~\cite{PetrenkoNR16,Nguena-TimoPR18} represents the fault domain for a DFSM specification with a NFSM. Each DFSM in the domain represents a version of the specification seeded with faults. The work addresses the problem of generating a test set~\cite{PetrenkoNR16} or a single test~\cite{Nguena-TimoPR18} for distinguishing a the specification from the other DFSMs. In this paper  we address a different concern, which is selecting a yet unknown oracle (specification) from a  set of candidate oracles.

In~\cite{TimoPR19}, experts play the role of an ultimate oracle to select one precise oracle from an imprecise oracle. The experts are requested to evaluate pairs of responses produced from too many pairs of  candidate oracles. In the current work, candidate oracles having produced unexpected responses are neither analysed, nor compared to the others.  The mining approach developed in this paper is clearly more efficient than the one in~\cite{TimoPR19}.

\section{Concluding remarks}\label{sec:conclusion}
We have presented an approach to  mining a precise  oracle from an imprecise one defining a set of candidate  oracles. Precise  oracles are represented with DFSMs whereas NFSMs represent imprecise oracles. We compactly encoded candidate precise oracles with Boolean formulas. We presented a method of reducing the imprecise oracle for efficient computation of plausible response sets. The proposed approach takes advantage of the efficiency of existing solvers and the reduction of the imprecise oracle for efficient search of distinguishable  precise oracles, test generation. It requests experts to select one correct response per test.  The experimental results have demonstrated that few tests and few response sets are needed for mining the proper  precise oracle from many candidate precise oracles. This indicates that the number of experts' interventions is reasonable and the approach is applicable. 

\par We plan to lift the proposed approach for mining extended finite state machines which are also used to represent test oracles. We also plan investigating automatic construction of imprecise oracles from system requirements, e.g., by modifying machine learning-based translation procedures or investigating mutation operators to be applied on generated "incorrect" oracles.

\section*{Acknowledgement}
This work was partially supported by MEI (Minist\`ere de l'\'Economie et Innovation) of Gouvernement du Qu\'ebec. 
The author would like to thank Dr. Alexandre Petrenko and anonymous reviewers for their valuable comments.

\bibliographystyle{splncs04}
\bibliography{reference}

\begin{thebibliography}{10}
\providecommand{\url}[1]{\texttt{#1}}
\providecommand{\urlprefix}{URL }
\providecommand{\doi}[1]{https://doi.org/#1}

\bibitem{DBLP:journals/iandc/Angluin87}
Angluin, D.: Learning regular sets from queries and counterexamples. Inf.
  Comput.  \textbf{75}(2),  87--106 (1987)

\bibitem{Barr2015}
Barr, E.T., Harman, M., McMinn, P., Shahbaz, M., Yoo, S.: The oracle problem in
  software testing: A survey. IEEE Transactions on Software Engineering
  \textbf{41}(5),  507--525 (May 2015)

\bibitem{brunello2019synthesis}
Brunello, A., Montanari, A., Reynolds, M.: Synthesis of ltl formulas from
  natural language texts: State of the art and research directions. In: 26th
  International Symposium on Temporal Representation and Reasoning (TIME 2019).
  Schloss Dagstuhl-Leibniz-Zentrum fuer Informatik (2019)

\bibitem{chen2020metamorphic}
Chen, T.Y., Cheung, S.C., Yiu, S.M.: Metamorphic testing: A new approach for
  generating next test cases. Tech. Rep. HKUST-CS98-01, Department of Computer
  Science, The Hong Kong University of Science and Technology (1998)

\bibitem{Chow78}
Chow, T.S.: Testing software design modeled by finite-state machines. {IEEE}
  Trans. Software Eng.  \textbf{4}(3),  178--187 (1978)

\bibitem{fantechi2003applications}
Fantechi, A., Gnesi, S., Lami, G., Maccari, A.: Applications of linguistic
  techniques for use case analysis. Requirements Engineering  \textbf{8}(3),
  161--170 (2003)

\bibitem{fantechi1994assisting}
Fantechi, A., Gnesi, S., Ristori, G., Carenini, M., Vanocchi, M., Moreschini,
  P.: Assisting requirement formalization by means of natural language
  translation. Formal Methods in System Design  \textbf{4}(3),  243--263 (1994)

\bibitem{Fraser:2013}
Fraser, G., Staats, M., McMinn, P., Arcuri, A., Padberg, F.: Does automated
  white-box test generation really help software testers? In: Proceedings of
  the 2013 International Symposium on Software Testing and Analysis. pp.
  291--301. ISSTA 2013, ACM, New York, NY, USA (2013)

\bibitem{Hieron2004}
Hierons, R.M.: Testing from a nondeterministic finite state machine using
  adaptive state counting. IEEE Transactions on Computers  \textbf{53}(10),
  1330--1342 (Oct 2004)

\bibitem{LeGoues:2012}
Le~Goues, C., Dewey-Vogt, M., Forrest, S., Weimer, W.: A systematic study of
  automated program repair: Fixing 55 out of 105 bugs for \$8 each. In:
  Proceedings of the 34th International Conference on Software Engineering. pp.
  3--13. ICSE '12, IEEE Press, Piscataway, NJ, USA (2012)

\bibitem{Lee96}
Lee, D., Yannakakis, M.: Principles and methods of testing finite state
  machines-a survey. Proceedings of the IEEE  \textbf{84}(8),  1090--1123 (Aug
  1996)

\bibitem{MavridouL18}
Mavridou, A., Laszka, A.: Designing secure ethereum smart contracts: {A} finite
  state machine based approach. In: Meiklejohn, S., Sako, K. (eds.) Financial
  Cryptography and Data Security - 22nd International Conference, {FC} 2018,
  Nieuwpoort, Cura{\c{c}}ao, February 26 - March 2, 2018, Revised Selected
  Papers. Lecture Notes in Computer Science, vol. 10957, pp. 523--540. Springer
  (2018)

\bibitem{Nguena-TimoPR18}
{Nguena Timo}, O., Petrenko, A., Ramesh, S.: Checking sequence generation for
  symbolic input/output fsms by constraint solving. In: Proceedings of 15th
  International Colloquium on Theoretical Aspects of Computing. Lecture Notes
  in Computer Science, vol. 11187, pp. 354--375. Springer (2018)

\bibitem{TimoPR19}
{Nguena Timo}, O., Petrenko, A., Ramesh, S.: Using imprecise test oracles
  modelled by {FSM}. In: 2019 {IEEE} International Conference on Software
  Testing, Verification and Validation Workshops, {ICST} Workshops 2019, Xi'an,
  China, April 22-23, 2019. pp. 32--39. {IEEE} (2019)

\bibitem{PetrenkoNR16}
Petrenko, A., {Nguena Timo}, O., Ramesh, S.: Multiple mutation testing from
  {FSM}. In: Albert, E., Lanese, I. (eds.) Proceedings of 36th {IFIP} {WG} 6.1
  International Conference on Formal Techniques for Distributed Objects,
  Components, and Systems. Lecture Notes in Computer Science, vol.~9688, pp.
  222--238. Springer (2016)

\bibitem{SahaK20}
Saha, P., Kanewala, U.: Improving the effectiveness of automatically generated
  test suites using metamorphic testing. In: {ICSE} '20: 42nd International
  Conference on Software Engineering, Workshops, Seoul, Republic of Korea, 27
  June - 19 July, 2020. pp. 418--419. {ACM} (2020)

\bibitem{Segura2016}
Segura, S., Fraser, G., Sanchez, A.B., Ruiz-Cortés, A.: A survey on
  metamorphic testing. IEEE Transactions on Software Engineering
  \textbf{42}(9),  805--824 (Sept 2016)

\bibitem{stahlberg2020neural}
Stahlberg, F.: Neural machine translation: A review. Journal of Artificial
  Intelligence Research  \textbf{69},  343--418 (2020)

\bibitem{Weimer:2009}
Weimer, W., Nguyen, T., Le~Goues, C., Forrest, S.: Automatically finding
  patches using genetic programming. In: Proceedings of the 31st International
  Conference on Software Engineering. pp. 364--374. ICSE '09, IEEE Computer
  Society, Washington, DC, USA (2009)

\bibitem{weyuker82}
Weyuker, E.J.: {On Testing Non-Testable Programs}. The Computer Journal
  \textbf{25}(4),  465--470 (Nov 1982)

\end{thebibliography}
\end{document}